\documentclass[preprint,12pt]{elsarticle}

\usepackage{amsfonts,amssymb,amsmath,amsthm,bm}
\usepackage{mathtools}
\usepackage{ dsfont }
\usepackage{color}

\newcommand{\Var} {\mbox{$\rm{Var}$\,}}
\newcommand{\Prob} {\mbox{$\rm{Prob}$\,}}

\newtheorem{lemma}{Lemma}
\newtheorem{theorem}{Theorem}
\newtheorem{corollary}{Corollary}
\newtheorem{remark}{Remark}

\newtheorem{proposition}{Proposition}
\usepackage{lineno}

\begin{document}

\begin{frontmatter}



\title{Coalescence in the diffusion limit of a Bienaym\'e-Galton-Watson branching process}

\author[label1]{Conrad J.\ Burden}
\ead{conrad.burden@anu.edu.au}
\author[label1,label2]{Albert C.\ Soewongsono}
\ead{albertchristian.soewongsono@utas.edu.au}
\address[label1]{Mathematical Sciences Institute, Australian National University, Canberra, Australia}
\address[label2]{School of Physical Sciences, University of Tasmania, Hobart, Australia.}

\begin{abstract}
We consider the problem of estimating the elapsed time since the most recent common ancestor of a finite random sample 
drawn from a population which has evolved through a Bienaym\'e-Galton-Watson branching process.  More specifically, we 
are interested in the diffusion limit appropriate to a supercritical process in the near-critical limit evolving over a large number of time steps. 
Our approach differs from earlier analyses in that we assume the only known information is the mean and variance of the number 
of offspring per parent, the observed total population size at the time of sampling, and the size of the sample.   
We obtain a formula for the probability that a finite random sample of the population is descended from a single ancestor in 
the initial population, and derive a confidence interval for the initial population size in terms of the final population size and the 
time since initiating the process. We also determine a joint likelihood surface from which confidence regions can be determined for simultaneously 
estimating two parameters, (1)~the population size at the time of the most recent common ancestor, and (2)~the time elapsed since the existence of 
the most recent common ancestor.
\end{abstract}

\end{frontmatter}


%
%
\section{Introduction}
\label{sec:Introduction}

Suppose one observes the current size of a population which has been evolving through a Bienaym\'e-Galton-Watson (BGW) branching process 
since some unknown time in the past.  Is it possible to infer the time that has elapsed since the most recent common ancestor (MRCA) of the 
current population, and of a random sample taken from the current population?  

This paper is a continuation of an earlier paper \citep{burden2016genetic} which gives a partial answer to this question in the form of a joint 
estimate of the elapsed time since coalescence to the MRCA and the corresponding size of the entire population at the time of the MRCA.  
Our new results extend the previous calculation to the MRCA of finite random samples of the population, and to a proposed procedure for obtaining 
confidence regions in the joint parameter space of the time since coalescence and the population size at the time of coalescence.  

There have been many studies of coalescence in BGW process \citep{lambert2013coalescent,pardoux2016probabilistic,grosjean2018genealogy}, 
specifically for critical \citep{durrett1978genealogy,athreya2012coalescence} sub-critical \citep{lambert2003coalescence,le2014coalescence} 
and supercritical processes.  In the current paper we concentrate on supercritical processes in the near-critical limit, or, equivalently, the continuous-time diffusion limit.  Closest to our treatment is that of \cite{o1995genealogy}, who derives a formula (with a minor correction in \cite{kimmel2015branching}) 
for the coalescent time for sample of size $n = 2$ as a fraction of the time since initialisation of the BGW process with a known initial population size.  
Numerical simulation of an ensemble of BGW processes initiated from a single founder by \cite{cyran2010alternatives} is consistent with this formula.  
The ultimate aim of Cyran and Kimmel's paper was to compare the effectiveness of various population genetics models in dating the time since 
mitochondrial Eve (mtE).  More recently O'Connell's formula has been generalised to samples of size $n > 2$ by \cite{harris2019coalescent}.  

Our approach differs from previous treatments in that we initialise the BGW process at the (a priori unknown) time of coalescence of a random 
sample of the current population, treat the time since coalescence and initial population size as parameters to be estimated, and determine confidence 
regions in the two-dimensional parameter space defined by these two parameters.  Central to our approach is the solution by \cite{feller1951diffusion} to 
the forward Kolmogorov equation defining the diffusion limit of a BGW process. This approach is closer in spirit to the intention of Cyran and Kimmel's 
aim of determining the time since mtE without prior knowledge of the time at which the process was initiated from a single founding individual.  
It also enables a direct comparison with the analogous solution by \cite{Slatkin:1991uq} to the problem of determining the distribution of the time since 
coalescence of a sample of size 2 in a Wright-Fisher (WF) population with exponential growth.  

The main results of the paper are stated in three theorems.  Theorem~\ref{theorem1} in Section~\ref{sec:CoalescenceOfN} gives a formula for the probability 
that a uniform random sample is descended from a single individual in the initial population of a BGW process, given the initial and final population sizes 
and the time since the process was initiated.  Theorem~\ref{theorem2} in Section~\ref{sec:ConfidenceInterval} gives confidence intervals for the unknown 
initial population size in a BGW process when the final population size and time since initiation are known.  These two theorems enable us to map out 
the likely initial conditions in terms of two scaled parameters arising from the diffusion limit; (1) a scaled population size at the time of initialisation of the 
process, denoted $\kappa_0$, and (2) the scaled time elapsed from initialisation of the process to observation of the current population, denoted $s$.  
Theorem~\ref{theorem3} in 
Section~\ref{sec:AsymptoticBehaviour}, which is contingent on a technical conjecture stated in the Section~\ref{sec:LikelihoodSurface}, 
gives a universal asymptotic density function for the likely time since coalescence of a sample in the limit of large scaled population sizes.   

The layout of this paper is as follows.  The diffusion limit of a BGW process, notational conventions employed in subsequent sections, and a derivation 
of Feller's solution for a specified initial population are detailed in Section~\ref{sec:DiffusionLimit}.  Arguments leading to Theorems~\ref{theorem1} 
and \ref{theorem2} and the proofs of these theorems are in Sections~\ref{sec:CoalescenceOfN} and \ref{sec:ConfidenceInterval} respectively.  
Section~\ref{sec:CoalescenceOfN} is a generalisation of analogous derivations in \citep{burden2016genetic}, which concentrated on the entire population 
as opposed to a finite random sample.   
Section~\ref{sec:LikelihoodSurface} is devoted to deriving a function ${\cal L}(s, \kappa_0; \kappa)$ defined over the $(s, \kappa_0)$,
the contours of which are confidence regions for the scaled time since coalescence of a sample and the scaled population size at the time of coalescence.  
Here $\kappa$ is the scaled final population size at the time of sampling.  By integrating out the initial population size $\kappa_0$, a comparison is made 
with the distributions derived by \cite{Slatkin:1991uq} for a sample of size $n = 2$ taken from an exponentially growing WF population. 
The asymptotic $\kappa \rightarrow \infty$ limit is analysed in Section~\ref{sec:AsymptoticBehaviour}, and this enables us to compare our analytic formulae with a numerical 
simulation in Section~\ref{sec:NumericalSimulation}.  Results are discussed and conclusions drawn in Section~\ref{sec:Conclusions}.  
A technical appendix is devoted to details of calculations arising in Section~\ref{sec:LikelihoodSurface}.  

%
%
\section{The diffusion limit of a Bienaym\'e-Galton-Watson (BGW) branching process}
\label{sec:DiffusionLimit}

Consider a Bienaym\'e-Galton-Watson (BGW) branching process consisting of a population of $M(t)$ haploid individuals reproducing in discrete, non-overlapping 
generations $t = 0, 1, 2, \ldots$, with initial population size $M(0) = m_0$.  The numbers of 
offspring per individual at each time step $t$ are identically and independently distributed (i.i.d.) random variables $S_\alpha$, 
$\alpha = 1, \ldots, M(t)$, whose common distribution has mean and variance 
\begin{equation}
\mathbb{E}(S_\alpha) = \lambda, \qquad \Var(S_\alpha) = \sigma^2,  \label{lambdaSigmaDef}
\end{equation}
and finite moments to all higher orders.  Thus 
\begin{equation}	\label{eq:MAtTPlus1}
M(t + 1) = \sum_{\alpha = 1}^{M(t)} S_\alpha, \quad i = 1, 2.  
\end{equation}

We are interested in the supercritical case of the diffusion limit studied by \citet{feller1951diffusion} in which the initial population $m_0$ becomes infinite 
while the growth rate simultaneously approaches the critical value $\lambda \rightarrow 1$ in such a way that the parameter 
\begin{equation}
\kappa_0 = \frac{2m_0 \log \lambda}{\sigma^2},  \label{kappa0Def}
\end{equation}
remains fixed.  As argued in \citet{burden2016genetic}, this parameter has a straightforward physical meaning:  
Suppose the initial population is divided into two subpopulations of 
roughly equal size.  If the initial population size $m_0$ is such that $\kappa_0 >> 1$, then with high probability descendant lineages of both subpopulations will survive 
as $t \rightarrow \infty$.  On the other hand, if $\kappa_0 << 1$ and the population does not go extinct, the entire surviving population as $t \rightarrow \infty$ will, 
with high probability, be descended from an individual in one of the two subpopulations.  Thus $\kappa_0 \approx 1$ sets a scale for inferring coalescent events.  
With this in mind, in the following lemma we introduce continuum quantities $K(s)$ and $s$ adapted from \citet[][Eq.~(36)]{burden2016genetic}.  

%
\begin{lemma} 
Define a continuum time $s$ and continuous random population sizes $K(s)$ by 
\begin{equation}	\label{sAndKDefns}
s = t\log\lambda, \quad K(s) = \frac{2 M(t)\log \lambda}{\sigma^2}.  
\end{equation}
The corresponding forward Kolmogorov equation for the density $f_K(\kappa; s)$ in the limit $\log\lambda \rightarrow 0$ and 
$m_0 \rightarrow \infty$ such that $\kappa_0$ defined by Eq,~(\ref{kappa0Def}) remains fixed is  
\begin{equation}	\label{fwdKolmog} 
\frac{\partial f_K(\kappa; s)}{\partial s} = - \frac{\partial(\kappa f_K(\kappa; s))}{\partial \kappa} + \frac{\partial^2(\kappa f_K(\kappa; s))}{\partial \kappa^2},
\end{equation}
where 
\begin{equation}
f_K(\kappa; s)\,d\kappa = \Prob(\kappa \le K(s) < \kappa + d\kappa). 
\end{equation}  
\end{lemma}
\begin{proof}
The derivation follows the formal method given in \citet[][Chapter~4]{Ewens:2004kx} for obtaining a forward-Kolmogorov equation.  
Equating an increment of continuum to one time step in the discrete model, set  
\begin{equation}	\label{sAndKDef}
\begin{split}
\delta s &= \log\lambda \\
\delta K(s) & = \frac{2\log\lambda}{\sigma^2} (M(t + 1) - M(t)).  
\end{split}
\end{equation}
In general, if 
\begin{equation}	 \label{EdeltaKGeneralCase}
\begin{split}
\mathbb{E}(\delta K(s) \mid K(s) = \kappa) & =  a(\kappa) \delta s + o(\delta s), \\
\mathbb{E}(\delta K(s)^2 \mid K(s) = \kappa) & =  b(\kappa) \delta s + o(\delta s), \\
\mathbb{E}(\delta K(s)^k \mid K(s) = \kappa) & =  o(\delta s), \quad k\ge3,  
\end{split}
\end{equation}
for finite functions $a(\kappa)$ and $b(\kappa)$ as $\delta s \rightarrow 0$, then the forward Kolmogorov equation takes the form
\begin{equation}
 \frac{\partial f_K(\kappa; s)}{\partial s} = - \frac{\partial}{\partial \kappa} \left(a(\kappa) f_K(\kappa; s)\right) 
 									+ \frac{1}{2} \frac{\partial^2}{\partial \kappa^2} \left(b(\kappa) f_K(\kappa; s)\right). 
 		\label{generalForwardKolmog}
\end{equation}
From Eqs.~(\ref{lambdaSigmaDef}), (\ref{eq:MAtTPlus1}) and (\ref{sAndKDef}), and setting $\kappa =  2m\log\lambda/\sigma^2$ one obtains 
\begin{eqnarray}
\mathbb{E}(\delta K(s) \mid K(s) = \kappa) & =  & \frac{2\log\lambda}{\sigma^2}\left(\mathbb{E}(M(t + 1)\mid M(t)=m) - m \right) \nonumber \\
	& = & \frac{2m\log\lambda}{\sigma^2} (\lambda - 1) \nonumber \\
	& = & \kappa \delta s + {\mathcal O}(\delta s^2), 
\end{eqnarray}
and 
\begin{eqnarray}
\mathbb{E}(\delta K(s)^2 \mid K(s) = \kappa) & =  &  \Var(\delta K(s) \mid K(s) = \kappa) + \mathbb{E}(\delta K(s) \mid K(s) = \kappa)^2 \nonumber \\
	& = & \left(\frac{2\log\lambda}{\sigma^2}\right)^2 \Var\left(M(t + 1)\mid M(t)=m\right) + {\mathcal O}(\delta s^2) \nonumber \\
	& = & 2\kappa \log\lambda  + {\mathcal O}(\delta s^2) \nonumber \\
	& = & 2\kappa \delta s  + {\mathcal O}(\delta s^2).  
\end{eqnarray}
Thus $a(\kappa) = \kappa$, $b(\kappa) = 2\kappa$, and Eq.~(\ref{fwdKolmog}) follows.  
\end{proof}
%
Note that the continuum scaling defined by Eqs.~(\ref{kappa0Def}) and (\ref{sAndKDefns}) is deliberately chosen so that all parameters defining the 
process are subsumed into a single parameter $\kappa_0$ which is present in the diffusion limit only via the initial conditions, and that the differential 
equation itself is parameter-free.  
However, while appropriate for current problem of exploring coalescence in the near-supercritical case $\lambda > 1$, the scaling is not appropriate  
appropriate for the near-subcritical case $\lambda < 1$ as the subsequent analysis assumes $\kappa_0 > 0$ and that the density $f_K(\kappa)$ 
is non-zero only for $\kappa \ge 0$.  

The solution to the forward Kolmogorov equation for a BGW process was first solved, albeit with a different continuum scaling, 
by \citet{feller1951diffusion, feller1951two} by means of a Laplace transform.  Summaries of Feller's method of solution can be found in \citet[][p236]{Cox78} 
and \citet{burden2018mutation}.  For completeness, the following lemma derives the solution for a supercritical Feller diffusion using the continuum scaling 
and notation of the current paper.  

%
\begin{lemma}
The solution solution to the forward Kolmogorov equation, Eq.~(\ref{fwdKolmog}), corresponding to an initial scaled population $K(0) = \kappa_0 > 0$ is
\begin{eqnarray}	\label{KDensity}
\lefteqn{f_K(\kappa; s | \kappa_0)  =   \delta(\kappa) \exp\left( - \frac{\kappa_0}{1 - e^{-s}}\right) + }\nonumber \\
	&  & \frac{1}{e^s - 1} \left(\frac{\kappa_0 e^s}{\kappa} \right)^\frac{1}{2} \exp \left\{ - \frac{\kappa_0 e^s + \kappa}{e^s - 1} \right\} 
							I_1\left(\frac{2(\kappa \kappa_0 e^s)^{\frac{1}{2}}}{e^s -1} \right), \qquad 0 \le \kappa < \infty, \nonumber \\
\end{eqnarray}
where $\delta(\cdot)$ is the Dirac delta function and $I_1(\cdot)$ is a modified Bessel function of the first kind.  In particular, 
\begin{equation}	\label{fKInitialCondition}
f_{K}(\kappa; 0 | \kappa_0) = \delta(\kappa - \kappa_0).
\end{equation}
\end{lemma}
\begin{proof}
Consider Laplace transform 
\begin{equation}	\label{LTofFKeqn}
\phi(\theta, s) = \int_0^\infty e^{-\theta \kappa} f_{K}(\kappa; s | \kappa_0) d\kappa. 
\end{equation}
Applying the Laplace transform to both sides of Eq.~(\ref{fwdKolmog}) and carrying through standard manipulations~\citep[][pp218, 219]{Cox78} gives 
\begin{equation}
\frac{\partial\phi(\theta, s)}{\partial s} + \theta(\theta - 1) \frac{\partial\phi(\theta, s)}{\partial \theta} = 0.  
\end{equation}
This is a first order linear partial differential equation for $\phi(\theta; s)$ whose initial boundary condition, corresponding to Eq.~(\ref{fKInitialCondition}), is
\begin{equation}	 \label{phiInitCond}
\phi(\theta, 0) = e^{-\theta \kappa_0}. 
\end{equation}
The solution is found by observing that $\phi(\theta; s)$ is constant along characteristic curves in the $(s,\theta)$-plane defined by 
\begin{equation}	 \label{charDefn}
0 = \frac{\partial\phi(\theta, s)}{\partial s} + \frac{d \theta}{ds} \frac{\partial\phi(\theta, s)}{\partial \theta}.
\end{equation}
Comparing Eqs.~(\ref{LTofFKeqn}) and (\ref{charDefn}) we see that the characteristic curves are solutions to the differential equation 
\begin{equation}
\frac{d\theta}{ds} = \theta(\theta - 1), 
\end{equation}
namely
\begin{equation} \label{thetaOfS}
\theta(s) = \frac{\theta(0)}{\theta(0) - (\theta(0) - 1)e^s}. 
\end{equation}
Thus, $\phi(s, \theta)$ can be found by tracing the characteristic curve back to the boundary point 
\begin{equation}
\theta(0) = \frac{\theta e^s}{1 + \theta(e^{s} - 1)}, 
\end{equation}
obtained by inverting Eq.~(\ref{thetaOfS}).  Combining this with the initial condition Eq.~(\ref{phiInitCond}) gives 
\begin{equation}	\label{phiSolution}
\phi(\theta, s) = \exp\left\{-\frac{\kappa_0\theta e^s}{1 + \theta(e^{s} - 1)}\right\}.
\end{equation}

It remains to invert the Laplace transform.   The inverse Laplace transform of 
\begin{equation}	\label{phiInTermsOfAB v}
\phi(\theta) = \exp\left\{-\frac{A \theta}{1 + B\theta}\right\}, 
\end{equation}
is derived in \citet[][p236, 250]{Cox78} as\footnote{The continuous part of the density stated in \citet[][p250]{Cox78} is missing a factor $\exp(-A/B)$.  
This appears to be an error.} 
\begin{equation}
f_K(\kappa) = \delta(\kappa)\exp\left(-\frac{A}{B}\right) + 
	\exp\left(-\frac{A + \kappa}{B}\right) \left(\frac{A}{\kappa B^2}\right)^\frac{1}{2} I_1\left(\frac{2(A\kappa)^\frac{1}{2}}{B}\right). 
\end{equation}
Setting $A = \kappa_0 e^s$ and $B = e^s - 1$ gives Eq.~(\ref{KDensity}).  
\end{proof}
%
\begin{remark}
In the first term of Eq.~(\ref{KDensity}) the coefficient of the delta-function is the probability of the population becoming extinct up to 
time $s$. See, for instance, \citet[][p206]{Bailey:1964fk}.  
\end{remark}
\begin{remark}
Eq.~(\ref{KDensity}) can be written in terms of a 1-parameter family of normalised density functions\footnote{Note that 
Eq.~(31) of \citet{burden2018mutation} defining the function $f_{\rm Feller}(\cdot)$ contains an error: the argument of the modified Bessel function 
should be $2\kappa_0 z^{\frac{1}{2}}$, not $2\kappa_0 z^{-\frac{1}{2}}$.
}  
\begin{equation}
f_{\rm Feller}(z; \xi) = \delta(z) e^{-\xi} + \xi z^{-\frac{1}{2}} e^{-\xi(1 + z)} I_1\left(2\xi z^\frac{1}{2} \right), 
\label{fFeller}
\end{equation}
where $z, \xi \ge 0$, as 
\begin{equation}
f_K(\kappa; s | \kappa_0)  =  \frac{1}{\kappa_0 e^s} f_{\rm Feller}\left(\frac{\kappa}{\kappa_0 e^s} ; \frac{\kappa_0}{1 - e^{-s}}\right).
\label{KDensityFromFFeller}
\end{equation}
\end{remark}
%
%
\section{Coalescence of a sample of $n$ individuals}
\label{sec:CoalescenceOfN}

Our aim is to estimate the time of the most recent common ancestor (MRCA) of a sample of $n$ individuals chosen randomly and uniformly from a 
population of scaled size $\kappa = (2m\log\lambda)/\sigma^2$, where $m$ is the population size at the ``current'' time $s = t \log\lambda$ since the 
coalescent event, which we set to be at time $0$.  
The population is assumed to have evolved via a BGW process with constant values of the parameters 
$\lambda$ and $\sigma^2$, which are assumed to be given.  Thus $\kappa$ is to be thought of as a known input parameter, while $s$ is an unknown
parameter which seek to estimate.  In the process, the scaled population size $\kappa_0 = 2m_0\log\lambda/\sigma^2$ at the time of the MRCA will also be estimated.  
The analysis follows the same reasoning as earlier work of \citet{burden2016genetic} on estimating the time of 
the MRCA of an entire population, known as mtE in the case of the existing human population.  
%
\begin{lemma}\label{lemma3}
Define ${\cal E}_n(\kappa_0, s)$ to be the event that a random uniform sample of $n$ individuals taken at time $s$ are descended from a 
single individual in the original population at time $0$, given that the initial scaled population size was $\kappa_0$.  Then the  
probability of the joint event that 
\begin{enumerate}
\item ${\cal E}_n(\kappa_0, s)$ happens; and 
\item the population size at time $s$ is in the range $[\kappa, \kappa + d\kappa)$  
\end{enumerate}
is 
\begin{eqnarray}	\label{ProbEnKGivenK0}		 
\lefteqn{\Prob({\cal E}_n(\kappa_0, s), K(s) \in [\kappa, \kappa + d\kappa) | K(0) = \kappa_0) } \nonumber \\
& \qquad\qquad= & \frac{\kappa_0 e^s}{(e^s - 1)^2} \frac{ n! I_n(2w)}{w^n} \exp\left(- \frac{\kappa_0 + \kappa e^{-s}}{1 - e^{-s}} \right) d\kappa.
\end{eqnarray}
where
\begin{equation}
w = \frac{(\kappa\kappa_0 e^s)^\frac{1}{2}}{e^s - 1} = \frac{(\kappa\kappa_0)^\frac{1}{2}}{2 \sinh \frac{1}{2}s}.
\label{wDef}
\end{equation}
\end{lemma}
\begin{proof}
Consider the population to be divided into two types with scaled population sizes $K_1(s)$ and $K_2(s)$, such that each type evolves as an independent BGW process
with initial conditions 
\begin{equation}
K_1(0) = x_0\kappa_0, \quad K_2(0) = (1 - x_0)\kappa_0, \quad 0 \le x_0 \le 1. 
\end{equation}
Introduce two random variables 
\begin{equation}
K(s) = K_1(s) + K_2(s), \qquad X(s) = \frac{K_1(s)}{K_1(s) + K_2(s)}.  
\label{KXDefn}
\end{equation}
Here $K(s)$ is the total population size and $X(s)$ is the fraction of the population which is of the first type at time $s$.  
Since $K_1(s)$ and $K_2(s)$ are independent, the joint density corresponding to the initial conditions $K(0) = \kappa_0$, $X(0) = x_0$ is 
\begin{equation}	\label{KXDensity}
f_{K,X}(\kappa, x; s | \kappa_0, x_0) = \kappa f_{K}(x \kappa; s | x_0 \kappa_0) f_{K}((1 - x) \kappa; s | (1 - x_0) \kappa_0),  
\end{equation}
defined over the range $0 \le \kappa < \infty$, $0 \le x \le 1$, $0 \le s < \infty$.  
The function $f_K(\cdot)$ is defined by Eq.~(\ref{KDensity}) and the factor of $\kappa$ arises from the Jacobian of the transformation Eq.~(\ref{KXDefn}).  

In the event ${\cal E}_n(\kappa_0, s)$, define descendants of the common ancestor of the sample of size $n$ to be of type-1, and members of 
the population who are not descended from this individual to be of type-2.  Setting $x_0 = 1/m_0$ in the density $f_{K,X}$ we have 
\begin{eqnarray}
\lefteqn{\Prob({\cal E}_n(\kappa_0, s), K(s) \in [\kappa, \kappa + d\kappa) | K(0) = \kappa_0) } \nonumber \\
	&& \qquad\qquad\qquad\qquad  = m_0 \int_0^1 x^n f_{K,X}\left(\kappa, x; s | \kappa_0, \frac{1}{m_0}\right) dx \, d\kappa. 
\label{ProbEnAndK}
\end{eqnarray}
The initial factor $m_0$ accounts for the $m_0$ possibilities for the common ancestor of the sample, the integral ranges over all possible fractions $x$ 
of the final population that can be descended from that ancestor, and the factor $x^n$ is the probability that all $n$ individuals in the sample are of type-1, 
given that $X(s) = x$.  

From Eqs.~(\ref{KDensity}) and (\ref{KXDensity}), and using the identity $\kappa \delta(\kappa x) = \delta(x)$ and the property of the modified Bessel 
function that $I_1(2w) = w + O(w^2)$ as $w \rightarrow 0$, we obtain 
\begin{eqnarray}
\lefteqn{f_{K,X}\left(\kappa, x; s | \kappa_0, \frac{1}{m_0}\right) = } \nonumber \\
& &  \left\{ \delta(x) \left(1 - \frac{\kappa_0}{m_0(1 - e^{-s})}\right) + \frac{1}{m_0} \frac{\kappa \kappa_0 e^s}{(e^s - 1)^2} \exp\left(\frac{-x \kappa}{e^s - 1} \right) + O\left( \frac{1}{m_0^2}\right) \right\} \times \nonumber \\
& & \quad \left\{ \frac{1}{\kappa} \delta(1 - x) \exp\left(\frac{-\kappa_0}{1 - e^{-s}}\right) + \right. \nonumber \\
& & \qquad 	\frac{1}{e^s - 1} \left( \frac{\kappa_0 e^s}{(1 - x)\kappa} \right)^\frac{1}{2} \exp \left( - \frac{\kappa_0 e^s + (1 - x)\kappa}{e^s - 1} \right)  \times \nonumber \\
& & \qquad  \quad		\left. I_1\left(\frac{2((1 - x)\kappa\kappa_0 e^s)^\frac{1}{2}}{e^s - 1} \right) + O\left( \frac{1}{m_0}\right) \right\}.
\label{fKXAt1OverM0}   
\end{eqnarray}

Substituting back into Eq.~(\ref{ProbEnAndK}), noting that the $\delta(x)$ term in the first factor of Eq.~(\ref{fKXAt1OverM0}) does not contribute because of the 
factor $x^n$ in the integrand, integrating out the $\delta(1 - x)$ in the second factor, and taking the limit $m_0 \rightarrow \infty$ gives, after some straightforward algebra, 
\begin{eqnarray}
\lefteqn{\Prob({\cal E}_n(\kappa_0, s), K(s) \in [\kappa, \kappa + d\kappa) | K(0) = \kappa_0) } \nonumber \\
& = & \frac{\kappa_0 e^s}{(e^s - 1)^2} \exp\left(- \frac{\kappa_0 + \kappa e^{-s}}{1 - e^{-s}} \right) 
		 \left( 1 + w \int_0^1 \frac{x^n}{(1 - x)^\frac{1}{2}} I_1\left(2w(1 - x)^\frac{1}{2} \right) dx \right) d\kappa,   \nonumber \\ 
\end{eqnarray}
where $w$ is defined by Eq.~(\ref{wDef}).  
The integral can be evaluated with the aid of Wolfram alpha~\citep{WolframAlpha} as
\begin{eqnarray}
\int_0^1 \frac{x^n}{(1 - x)^\frac{1}{2}} I_1\left(2w(1 - x)^\frac{1}{2} \right) dx & = & 
	2 \int_0^1 (1 - z^2)^n I_1(2wz) dz \nonumber \\
& = & \frac{1}{w} \left( _0F_1(;n + 1; w^2) - 1 \right) \nonumber \\     
& = & \frac{1}{w} \left( \frac{ n! I_n(2w)}{w^n} - 1 \right),     
\end{eqnarray}
from which the required result, Eq.~(\ref{ProbEnKGivenK0}), follows.  
\end{proof}
%
\begin{theorem}\label{theorem1}
Define $u_n(s, \kappa_0 | \kappa)$ to be the probability that the initial population of a BGW process contains an individual who is the single common 
ancestor of a random uniform sample of $n$ individuals taken at a later time $s$, given the initial and final population sizes are $\kappa_0$ and 
$\kappa$ respectively.  Then   
\begin{eqnarray}	\label{u_nResult}
u_n(s, \kappa_0 | \kappa) & := & \Prob\left({\cal E}_n(\kappa_0, s) | K(0) = \kappa_0, K(s) = \kappa \right) \nonumber \\
	& = & \frac{n!}{w^{n - 1}} \frac{I_n(2w)}{I_1(2w)}.  
\end{eqnarray}
where $w$ is defined by Eq.~(\ref{wDef}).
\end{theorem}
\begin{proof}
By definition, 
\begin{equation}	\label{uByDefn}
u_n(s, \kappa_0 | \kappa)  = \frac{\Prob({\cal E}_n(\kappa_0, s), K(s) \in [\kappa, \kappa + d\kappa) | K(0) = \kappa_0)}
  		{\Prob \left( K(s) \in [\kappa, \kappa + d\kappa) | K(0) = \kappa_0 \right)}.  
\end{equation}
The numerator follows from Lemma~\ref{lemma3}.  The denominator, obtained from the continuous part of Eq.~(\ref{KDensity}), is
\begin{eqnarray}	\label{ProbKEqualsKappa}
\lefteqn{\Prob \left( K(s) \in [\kappa, \kappa + d\kappa) | K(0) = \kappa_0 \right)} \nonumber \\
& = & \frac{1}{e^s - 1} \left(\frac{\kappa_0 e^s}{\kappa} \right)^\frac{1}{2} \exp \left\{ - \frac{\kappa_0 + \kappa e^{-s}}{1 - e^{-s}} \right\} 
							I_1(2w) d\kappa.  
\end{eqnarray}
Substituting Eqs.(\ref{ProbEnKGivenK0}) and (\ref{ProbKEqualsKappa}) into Eq.~(\ref{uByDefn}) yields Eq.~(\ref{u_nResult}).  
\end{proof}
%
\begin{corollary}\citep{burden2016genetic}
The probability that the entire population at time $s$ in a BGW process is descended from a single individual in the initial population at time 0, 
given the initial and final population sizes are $\kappa_0$ and $\kappa$ respectively, is  
\begin{equation}	\label{u_InfResult}
u_\infty(s, \kappa_0 | \kappa) = \frac{w}{I_1(2w)}.  
\end{equation}
\end{corollary}
\begin{proof}
From the ascending series for modified Bessel function \citep[p375 of][]{Abramowitz:1965sf} we have 
$I_n(2w) = (w^n/n!)(1 + O(1/n))$ as $n \rightarrow \infty$.  Taking the limit of Eq.~(\ref{u_nResult}) leads to Eq.~(\ref{u_InfResult}).  
\end{proof}

The interpretation of Theorem~\ref{theorem1} is illustrated by the black contours in Figure~\ref{fig:coalescentContours}. 
Figures~\ref{fig:coalescentContours}(a) and  \ref{fig:coalescentContours}(b) show a contour map of the density in Eq.~(\ref{u_nResult}) in the $s$-$\kappa_0$ 
plane for a random sample of size $n = 2$ taken from an observed final scaled population of size $\kappa = 3$ and $\kappa = 2000$ respectively.  
The map grades from a probability near 0 on the left edge of the plot corresponding to recent initialisation times to a probability 
near 1 at earlier times.  If the BGW process were initiated at a time $s$ in the past with a population of size $\kappa_0$, the contour passing through the 
point $(s, \kappa_0)$ gives the probability that the ancestral coalescence of the sample occurred after the initialisation of the process.  
Note that for larger values of $\kappa$ the coalescent event is fixed within a relatively narrow range.  

\begin{figure}[t]
\begin{center}
\centerline{\includegraphics[width=\textwidth]{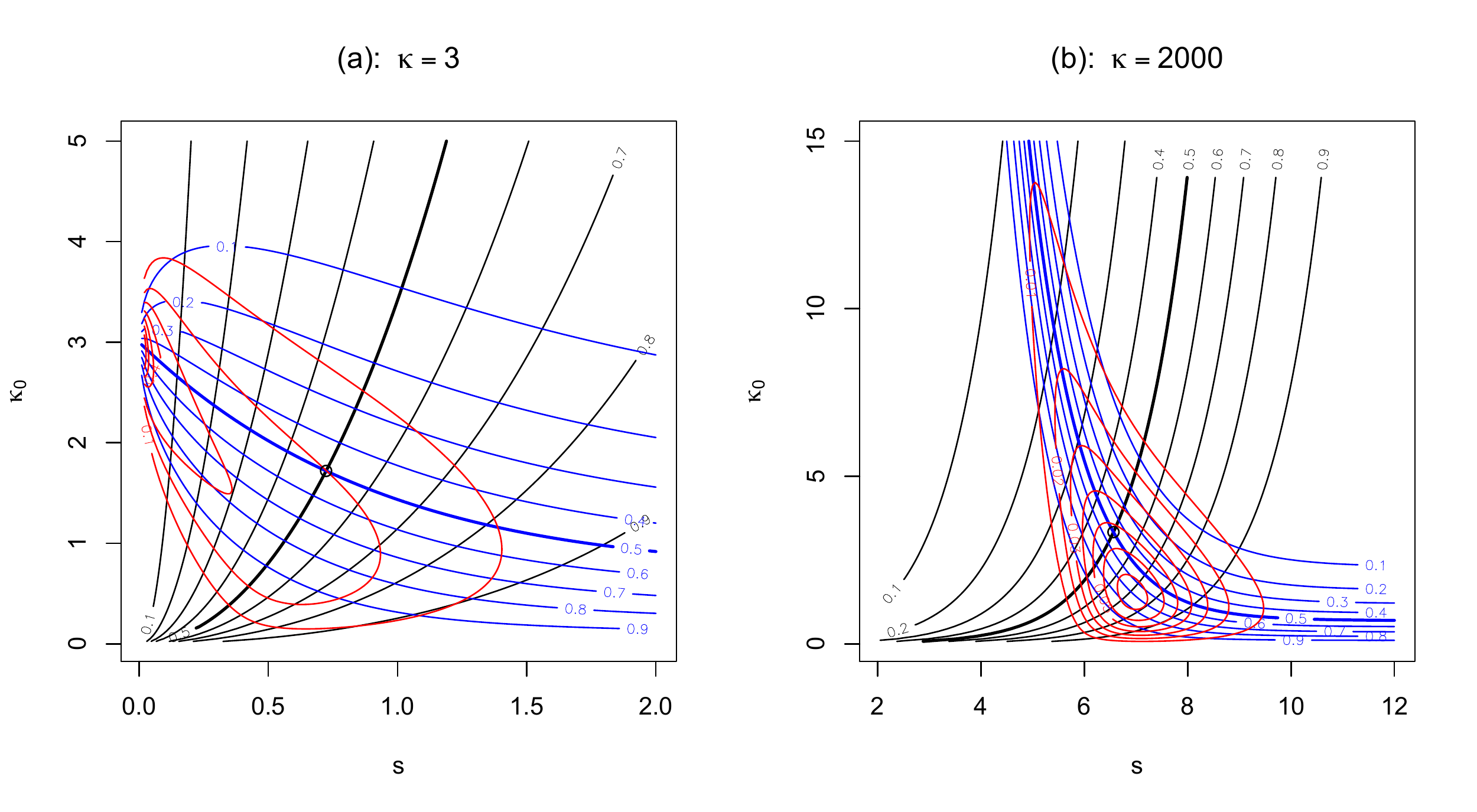}}
\caption{Contours (black lines) of the probability $u_n(s, \kappa_0 | \kappa)$ in Eq.~(\ref{u_nResult}) that a random sample of $n=2$ individuals 
in the ``current" population is descended from a single individual in the ``initial" population.  
Axes refer to the scaled time $s$ at which the current population is measured and the scaled initial population at time zero, $\kappa_0$.  
Superimposed in blue are contours of the likelihood function $v(s, \kappa_0 | \kappa)$ defined in  
Eq.~(\ref{vDefn}) as probability that the population at time $s$ will not exceed the observed population 
$\kappa$ for a given starting position in the plane.  
The red curves are the contours of the likelihood surface ${\cal L}(s, \kappa_0 ; \kappa)$ described in \S\ref{sec:LikelihoodSurface}. } 
\label{fig:coalescentContours}
\end{center}
\end{figure}

%
%
\section{Confidence interval for $\kappa_0$ given an observed $K(s)$}
\label{sec:ConfidenceInterval}

Superimposed in blue in Figure~\ref{fig:coalescentContours} are contours of the likelihood function that the population at time $s$ will not 
exceed the observed population $\kappa$ for a given starting population $\kappa_0$, 
\begin{eqnarray}	\label{vDefn}
v(s, \kappa_0 | \kappa) & := & \Prob(K(s) \le \kappa | K(0) = \kappa_0 ) \nonumber \\
	& = & \int_0^\kappa f_K(\eta; s |  \kappa_0) d\eta,  \label{vDefn}
\end{eqnarray}
where the function $f_K$ is defined by Eq.~(\ref{KDensity}).  
Note that the integrand includes the point mass at $\eta = 0$ which accounts for 
the possibility of extinction of the entire population.  The contour marked $0.1$ connects starting configurations for which there is a 10\% probability the 
population will undershoot the observed population at time $s$, and the contour marked $0.9$ connects starting configurations for which there is a 10\% 
chance of overshooting.  

The function $v(s, \kappa_0 | \kappa)$ enables us to establish confidence interval for estimating the initial population $K(0) = \kappa_0$ given that 
a population generated by a BGW process is observed to be of size $K(s)$ at a known time $s$.  This confidence interval is define in the following theorem.  
%
\begin{theorem} \label{theorem2}
Consider the inverse function $v^{-1}(p | s, \kappa)$ mapping the probability $p \in [0, 1]$ back to an initial condition $\kappa_0$ in 
the interval $[0, \infty)$, thus, 
\begin{eqnarray}
v^{-1}(v(s, \kappa_0 | \kappa) | s, \kappa) = \kappa_0.  
\label{vInvDefn}
\end{eqnarray}
Then if a BGW process initiated at time 0 with population $\kappa_0$ evolves to a population $K(s)$ at time $s$, 
for any $0 \le p_1 < p_2 \le 1$, the random interval 
\begin{equation}
{\cal I}_{\kappa_0}(p_1, p_2 | K(s)) := [v^{-1}(p_2 | s, K(s)), v^{-1}(p_1 | s, K(s))], 
\label{kappa0ConfInterval}
\end{equation}
has a probability $p_2 - p_1$ of containing the starting population $\kappa_0$. 
\end{theorem}
\begin{proof}
It is clear that $v^{-1}$ is in general a decreasing function of $p$, with $v^{-1}(0 | s, \kappa) = \infty$ 
and $v^{-1}(1 | s, \kappa) = 0$.  Then 
\begin{eqnarray}
\lefteqn{\Prob({\cal I}_{\kappa_0}(p_1, p_2 | K(s)) \ni \kappa_0)} \nonumber \\
& = & \Prob(v^{-1}(p_1 | s, K(s)) \ge \kappa_0 \ge v^{-1}(p_2 | s, K(s))\nonumber \\
& = & \Prob(p_1\le v(s, \kappa_0|K(s)) \le p_2)\nonumber \\
& = & p_2 - p_2. 
\end{eqnarray}
The second last line follows because $v^{-1}$ is a decreasing function of $p$, and the final line follows because the cumulative distribution 
$v(s, \kappa_0|K(s))$ is uniformly distributed on $[0, 1]$.  
\end{proof}

Thus, for instance, the blue contours marked $v = 0.1$ and $v = 0.9$ define the edges of an 80\% confidence interval on starting configurations $(s, \kappa_0)$ 
for the observed final population $\kappa$.  The small circle in each plot is a ``median estimate'' of the coalescence point which we define as 
the (scaled) starting time and population size for which there is a 50\% probability that the coalescent event is yet to occur, and a 50\% probability 
that the final population will overshoot the observed current observed population size.  
%
%
\section{Likelihood surface given an observed $K(s)$}
\label{sec:LikelihoodSurface}

Naturally one would like to devise a more intuitive description of the likely location of the ancestral coalescent point of a random sample than the 
median estimate defined in the preceding section.  Consider the following scenario.  A BGW process is initiated with an initial scaled population 
$\kappa_0^{\rm true}$, and a random sample on $n$ individuals taken from the descendant population at some later time $s^{\rm true}$.  
The process is further conditioned on the event that the MRCA of the sample coincides with the initial time.  This is in principle 
achieved by creating an arbitrarily large ensemble of processes, and rejecting those for which the required condition is not satisfied to within  
some small tolerance.  Our aim is to determine a function ${\cal L}(s, \kappa_0 ; \kappa)$ with the property that, for any subset $\Omega$ of the 
$(s, \kappa_0)$-plane and observed final population $K(s^{\rm true})=\kappa$,  
\begin{equation}	\label{confidenceIntegral}
\Prob\left((s^{\rm true}, \kappa_0^{\rm true}) \in \Omega\right) = \int_\Omega {\cal L}\left(s, \kappa_0 ; \kappa \right) ds\,d\kappa_0. 
\end{equation}

We will refer to the function ${\cal L}(s, \kappa_0 ; \kappa)$ as a ``likelihood surface".   This surface can be considered as a tool for generating 
confidence regions in the sense that, by tuning $\Omega$ so that the integral in Eq.~(\ref{confidenceIntegral}) has a value 0.95 for instance, 
one creates a $95\%$ confidence region for the point $(s^{\rm true}, \kappa_0^{\rm true})$.  

Subject to a conjecture which will be elucidated in the proof below, we have the following proposition: 
%
\begin{proposition}	\label{proposition1}
The function 
\begin{equation}	\label{LikelihoodDensity}	    
{\cal L}(s, \kappa_0 ; \kappa) = 
\left| \frac{\partial u_n(s, \kappa_0 | \kappa)}{\partial s} \frac{\partial v(s, \kappa_0 | \kappa)}{\partial \kappa_0} - 
        \frac{\partial u_n(s, \kappa_0 | \kappa)}{\partial \kappa_0} \frac{\partial v(s, \kappa_0 | \kappa)}{\partial s} \right|,
\end{equation}
where $K = \kappa$ is the observed current population, the function $u_n(s, \kappa_0 | \kappa)$, defined in Theorem~\ref{theorem1}, 
and $v(s, \kappa_0 | \kappa)$ is defined by Eq.~(\ref{vDefn}) satisfies Eq.~(\ref{confidenceIntegral}).  
\end{proposition}

\begin{figure}[t]
\begin{center}
\centerline{\includegraphics[width=\textwidth]{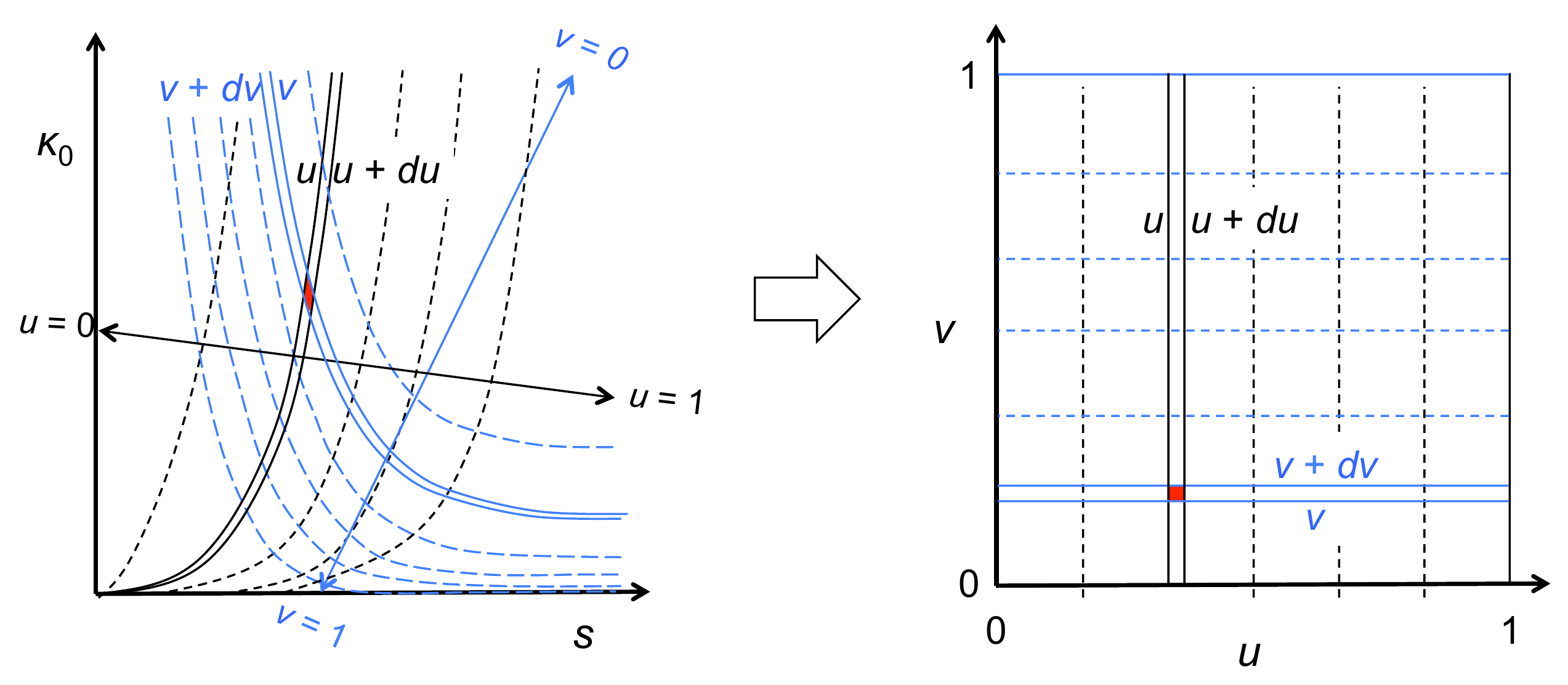}}
\caption{Mapping of the $(s, \kappa_0)$-plane to the $(u, v)$-plane.} 
\label{fig:MappingToUVPlane}
\end{center}
\end{figure}

\begin{proof}
Consider two neighbouring contours in the $(s, \kappa_0)$-plane of the function $u_n$ given by Eq.~(\ref{u_nResult}) or (\ref{u_InfResult}), 
taking values $u$ and $u + du$ respectively, as shown in the left panel of Fig.~\ref{fig:MappingToUVPlane}.  Recall that $u_n(s, \kappa_0 | \kappa)$ 
is the probability that a BGW process initiated at any point on the first contour will result in a sample of size $n$ at time $s$ having a single ancestor 
in the initial population, conditional on $K(0) = \kappa_0$ and $K(s) = \kappa$. Similarly $u(s, \kappa_0 | \kappa) + du(s, \kappa_0 | \kappa)$ is the 
corresponding probability for the second contour.   Then for fixed values of $\kappa_0$ and $\kappa$, the interval 
\begin{equation}
{\cal I}_s(u, u + du | \kappa_0, \kappa) := \{s : u <u_n(s, \kappa_0 | \kappa) < u + du \}, 
\end{equation}
has a probability $du$ of containing the time $s^{\rm true}$ since the MRCA of the sample.  

Furthermore, by Theorem~\ref{theorem2}, at a given value of $s$ the interval delimited by the contours 
$v$ and $v + dv$ shown in the left panel of Fig.~\ref{fig:MappingToUVPlane}, namely 
\begin{equation}
{\cal I}_{\kappa_0}(v, v + dv | s, \kappa) := \{\kappa_0 : v <v(s, \kappa_0 | \kappa) < v + dv \}, 
\end{equation}
has a probability $dv$ of containing the initial population $\kappa_0^{\rm true}$, given an observed final population size $\kappa$.  

These two intervals are random inasmuch as they depend on the final population size $K(s^{\rm true}) | (K(0) = \kappa_0^{\rm true})$ 
of the conditioned BGW process described above.  We {\bf conjecture} that the two intervals are independent, and that therefore the 
probability that the intersection of the two corresponding regions in the $(s, \kappa_0)$-plane contains the point 
$(s^{\rm true}, \kappa_0^{\rm true})$ is $du\,dv$.  
Equivalently, we conjecture that if the $(s, \kappa_0)$-plane is mapped onto the $(u, v)$-plane, as in Fig.~\ref{fig:MappingToUVPlane}, 
then the required function ${\cal L}(s, \kappa_0 ; \kappa)$ maps to a uniform density over $[0, 1] \times [0, 1]$.  
The inverse mapping to the $(s, \kappa_0)$-plane is effected by the Jacobian of the transformation, leading to Eq.~(\ref{LikelihoodDensity}).  
\end{proof}

Contours of the likelihood surface ${\cal L}(s, \kappa_0 ; \kappa)$ are marked in red in Figure~\ref{fig:coalescentContours} for $\kappa = 3$ and 
$\kappa = 2000$.  Details of the calculation of the required derivatives are set out in \ref{sec:DerivatesUAndV}.  
These contour lines have the property that they bound a minimal area confidence region in the $(s, \kappa_0)$-plane for a given integrated likelihood.  
Immediately noticeable is that in both cases the confidence regions are extremely skewed, and the median estimate defined at the end of the previous 
section is well separated from the maximum of the likelihood surface.  

\begin{figure}[t]
\begin{center}
\centerline{\includegraphics[width=\textwidth]{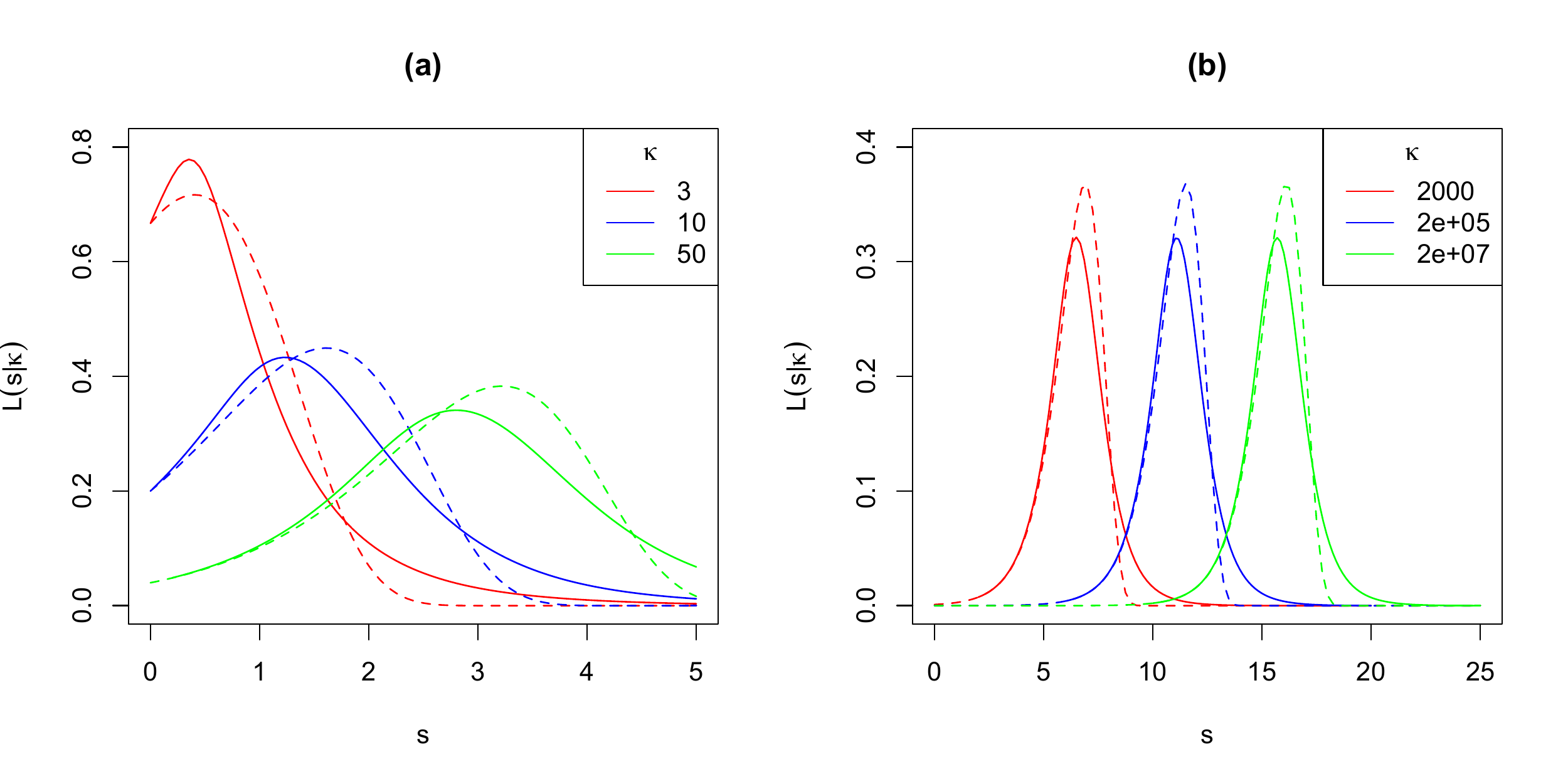}}
\caption{The marginal likelihood ${\cal L}(s;\kappa)$ of the time since coalescence for $n = 2$ individuals sampled from a population of current size 
$\kappa$ for a BGW process (Eq.~(\ref{MarginalLikelihood}), solid curves) and the probability density $f_{S}^{\rm SH}(s | \kappa)$ of the time since 
coalescence for an exponentially growing WF population (Eq.~(\ref{fSSHDefn}), dashed curves).  } 
\label{fig:marginalSLhood}
\end{center}
\end{figure}

The marginal likelihood of the scaled time since coalescence, which we define from Eq.~(\ref{LikelihoodDensity}) as 
\begin{equation}	\label{MarginalLikelihood}
{\cal L}(s ; \kappa) = \int_0^\infty {\cal L}(s, \kappa_0 ; \kappa) d\kappa_0,  
\end{equation}
is plotted in Fig.~\ref{fig:marginalSLhood} for a sample of $n = 2$ individuals and various values of the currently observed scaled population $\kappa$.  
The numerical integration is tricky to evaluate for large values of $\kappa$ and small values of $s$ because of the difficulty in numerically evaluating 
$\partial v/\partial\kappa_0$ along the ridge in the joint density 
${\cal L}(s, \kappa_0 ; \kappa)$ running up to the point $(s, \kappa_0) = (0, \kappa)$ (see Fig.~\ref{fig:coalescentContours}(b), for instance).  
To produce the plots in Fig.~\ref{fig:marginalSLhood}(b) we avoided this problem by integrating by parts to obtain 
\begin{equation}
{\cal L}(s ; \kappa) = \int_0^\infty  
		\left[ \frac{\partial u_n(s, \kappa_0 | \kappa)}{\partial \kappa_0} \frac{\partial v(s, \kappa_0 | \kappa)}{\partial s}  + 
     			  \frac{\partial^2 u_n(s, \kappa_0 | \kappa)}{\partial s \partial \kappa_0} v(s, \kappa_0 | \kappa) \right]  d\kappa_0,
\label{marginalLIntByParts}
\end{equation}
where the second partial derivative of $u_n$ is given in \ref{sec:DerivatesUAndV}, Eq.~(\ref{du_nDsDkappa0}).

Since the marginal likelihood integrates to 1, a comparison can be made with the probability distribution of times since coalescence for a 
sample of size $n = 2$ taken from an exponentially 
growing WF population, as determined by \cite{Slatkin:1991uq}. Their distribution is quoted in terms of a single parameter $\alpha = N_0 r$ where $N_0$
is the final observed haploid population size and $r$ is the exponential growth rate.  In terms of our notation, the expected number of offspring 
per individual per generation is $\lambda = e^r$, and the variance of the number of offspring per individual per generation for a WF model in the 
diffusion limit is $\sigma^2 = 1 + O(1/M(t))$.  Thus the condition $K(s) = 2 M(t)\log\lambda/\sigma^2 = \kappa$ translates to $\kappa = 2\alpha$, 
and Slatkin and Hudson's Eq.~(6) translates to  
\begin{equation}
f_{S}^{\rm SH}(s | \kappa) = \frac{2 e^s}{\kappa} \exp\left(- \frac{2 (e^s - 1)}{\kappa}\right).  
\label{fSSHDefn}
\end{equation}
This distribution is plotted as dashed curves in Fig.~\ref{fig:marginalSLhood}.  Note that the parameters chosen in Fig.~\ref{fig:marginalSLhood}(b) 
match those of Fig.~4 in \cite{Slatkin:1991uq}.  

\begin{figure}[t]
\begin{center}
\centerline{\includegraphics[width=0.6\textwidth]{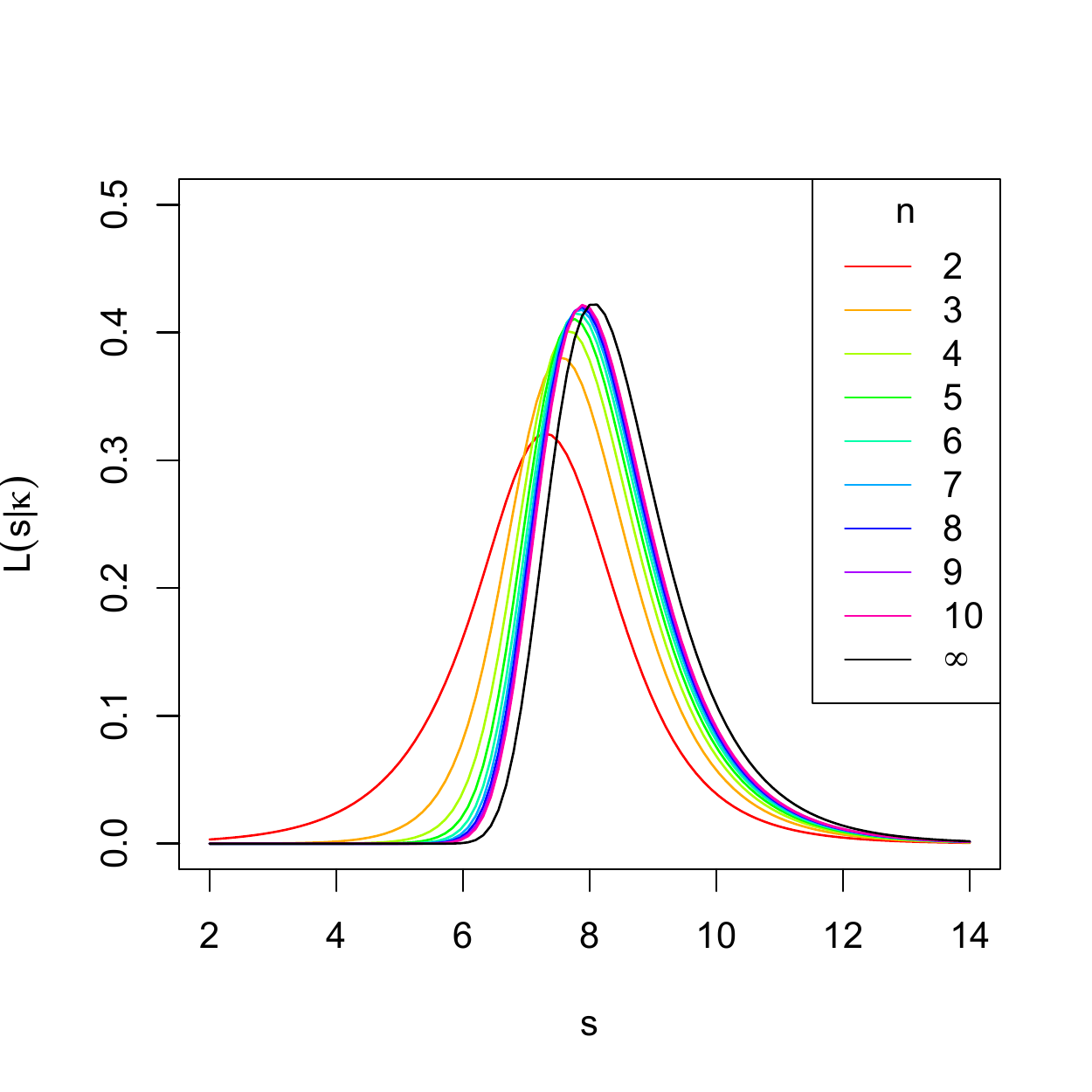}}
\caption{The marginal likelihood ${\cal L}(s;\kappa)$ of the time since coalescence for samples of varying sizes taken from a population of current size 
$\kappa = 4500$ for a BGW process.  } 
\label{fig:marginalSLhoodVaryingN}
\end{center}
\end{figure}

Figure~\ref{fig:marginalSLhoodVaryingN} shows the marginal likelihood of the time since coalescence for samples of size $n = 2, 3, \ldots$.  The 
choice of scaled current population, $\kappa = 4500$, matches the naive model of human population growth during the upper Paleolithic period employed 
by \citet{burden2016genetic} to illustrate an application of BGW processes to mtE.  In common with the predictions of the WF model 
for an exponentially growing population~\citep{Slatkin:1991uq}, we observe that MRCAs for samples of any size are likely to be restricted to a limited 
time range, consistent with a ``star'' shaped phylogeny.  
%
%
\section{Asymptotic behaviour for large $\kappa$}
\label{sec:AsymptoticBehaviour}

Plots of the marginal likelihood and of the Slatkin-Hudson density in Fig.~\ref{fig:marginalSLhood} suggest that these functions each converge to 
an asymptotic shape as $\kappa \rightarrow \infty$.  Specifically, by defining a shifted time scale 
\begin{equation}
s' = s - \log\kappa, 
\label{sPrimeDef}
\end{equation} 
one easily checks from Eq.~(\ref{fSSHDefn}) that 
\begin{equation}
f_S^{\rm SH}(s' + \log\kappa | \kappa) = 2e^{s'} e^{-2e^{s'}}\left(1 + O(\kappa^{-1})\right), 
\end{equation}
as $\kappa \rightarrow \infty$.  Thus the Slatkin-Hudson density rapidly takes on a shape related to the Gumbel distribution for large $\kappa$.  
As we next show, it turns out that each of the likelihood functions defined above for a BGW process also 
converges to an asymptotic shape depending only on $s'$ (and, where relevant, $\kappa_0$) for large $\kappa$.  

\begin{theorem}	\label{theorem3}
The functions $u_n$ and $v$ whose contours are plotted in Fig.~\ref{fig:coalescentContours}, and hence the likelihood surface Eq,(\ref{LikelihoodDensity}) 
and marginal likelihood Eq.~(\ref{MarginalLikelihood}), each take on a universal form as functions of $s'= s - \log\kappa$ and $\kappa_0$ as $\kappa \rightarrow \infty$.  
\end{theorem}
\begin{proof}
From Eqs.~(\ref{u_nResult}), (\ref{u_InfResult}) 
and (\ref{wDef}) we have that $u_n$ and $u_\infty$ are functions only of $w$ which behaves asymptotically as 
$(\kappa_0 e^{-s'})^\frac{1}{2} (1 + O(\kappa^{-1}))$.  It immediately follows that for large $\kappa$, 
\begin{equation}
u_n(s, \kappa_0 | \kappa) \approx \tilde{u}_n(s - \log\kappa, \kappa_0), \qquad 
u_\infty(s, \kappa_0 | \kappa) \approx \tilde{u}_\infty(s - \log\kappa, \kappa_0), \qquad 
\end{equation}
where 
\begin{equation}
\begin{split}
\tilde{u}_n(s', \kappa_0) &= \frac{n!}{w'^{n - 1}} \frac{I_n(2w')}{I_1(2w')}, \\   
\tilde{u}_\infty(s', \kappa_0) &= \frac{w'}{I_1(2w')},  \qquad
w' = (\kappa_0 e^{-s'})^\frac{1}{2}. 
\end{split}
\end{equation} 

The asymptotic form of $v$ is obtained from Eqs.~(\ref{vDefn}) and (\ref{KDensityFromFFeller}) by substituting Eq.~(\ref{sPrimeDef}) and discarding 
terms of $O(\kappa^{-1})$ to obtain 
\begin{equation}
v(s, \kappa_0 | \kappa) \approx \tilde{v}_n(s - \log\kappa, \kappa_0), 
\end{equation}
where
\begin{equation}
\tilde{v}_n(s', \kappa_0) = F_{\rm Feller}\left(\frac{1}{\kappa_0 e^{s'}}; \kappa_0 \right). 
\end{equation}
Here $F_{\rm Feller}(z, \kappa_0) = \int_{0-}^z f_{\rm Feller}(\zeta, \kappa_0) d\zeta$ is the cumulative distribution corresponding to the density defined by Eq.~(\ref{fFeller}).  
\end{proof}

\begin{figure}[t!]
\begin{center}
\centerline{\includegraphics[width=\textwidth]{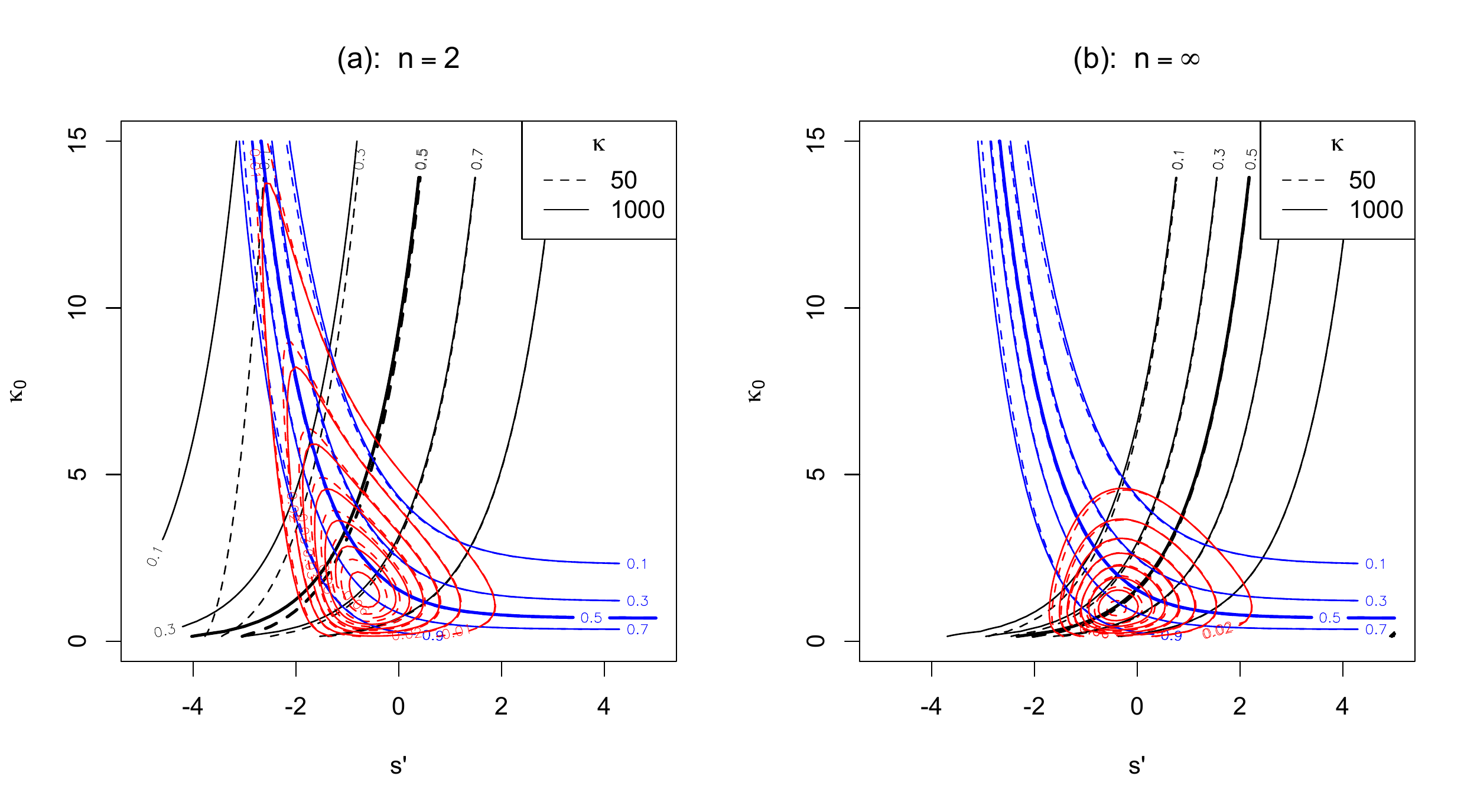}}
\caption{Contours of the probability $u_n(s' + \log \kappa, \kappa_0|\kappa)$ that a random sample of $n$ individuals is descended from a single individual (black lines), 
the probability $v(s' + \log \kappa, \kappa_0|\kappa)$ that the population at time $s$ will not exceed the observed population $\kappa$ (blue lines), and the corresponding 
likelihood surface (red contours) as functions of the shifted time $s'$ and currently observed population $\kappa_0$.   } 
\label{fig:coalescentContoursShifted}
\end{center}
\end{figure}

Figure~\ref{fig:coalescentContoursShifted} shows contours of the functions $u_2$, $u_\infty$, $v$, and contours of associated likelihood surfaces, 
${\cal L}(s' + \log\kappa, \kappa_0 ;\kappa)$ defined by Eq.~(\ref{LikelihoodDensity}),  for relatively large values  
$\kappa = 50$ and $1000$ of the observed current population.  
On the scale shown, the $\kappa = 1000$ contours are indistinguishable from those corresponding to the universal functions 
$\tilde{u}_n$,  $\tilde{u}_\infty$,  and $\tilde{v}$ (not shown).   Note that the universal functions are approached more rapidly for larger values of $n$.  
Contours of the limiting likelihood $\tilde{\cal L}(s', \kappa_0) := \lim_{\kappa \rightarrow \infty} {\cal L}(s' + \log\kappa, \kappa_0 ; \kappa)$ 
and the limiting marginal likelihood $\tilde{\cal L}(s') := \lim_{\kappa \rightarrow \infty} {\cal L}(s' + \log\kappa ; \kappa)$, calculated from 
$\tilde{u}_n(s', \kappa_0)$ and $\tilde{v}(s', \kappa_0)$ are plotted in Fig.~\ref{fig:AnalyseSimulatedGWData}(c) to (f).  We have observed that 
the numerically calculated limiting marginal likelihood for $n = \infty$, that is, the red curve in Figures~\ref{fig:AnalyseSimulatedGWData}(d) and (f), 
is a very close fit to a shifted Gumbel distribution.  So far we have been unable to verify this analytically.  

\begin{figure}[t!]
\begin{center}
\centerline{\includegraphics[width=0.82\textwidth]{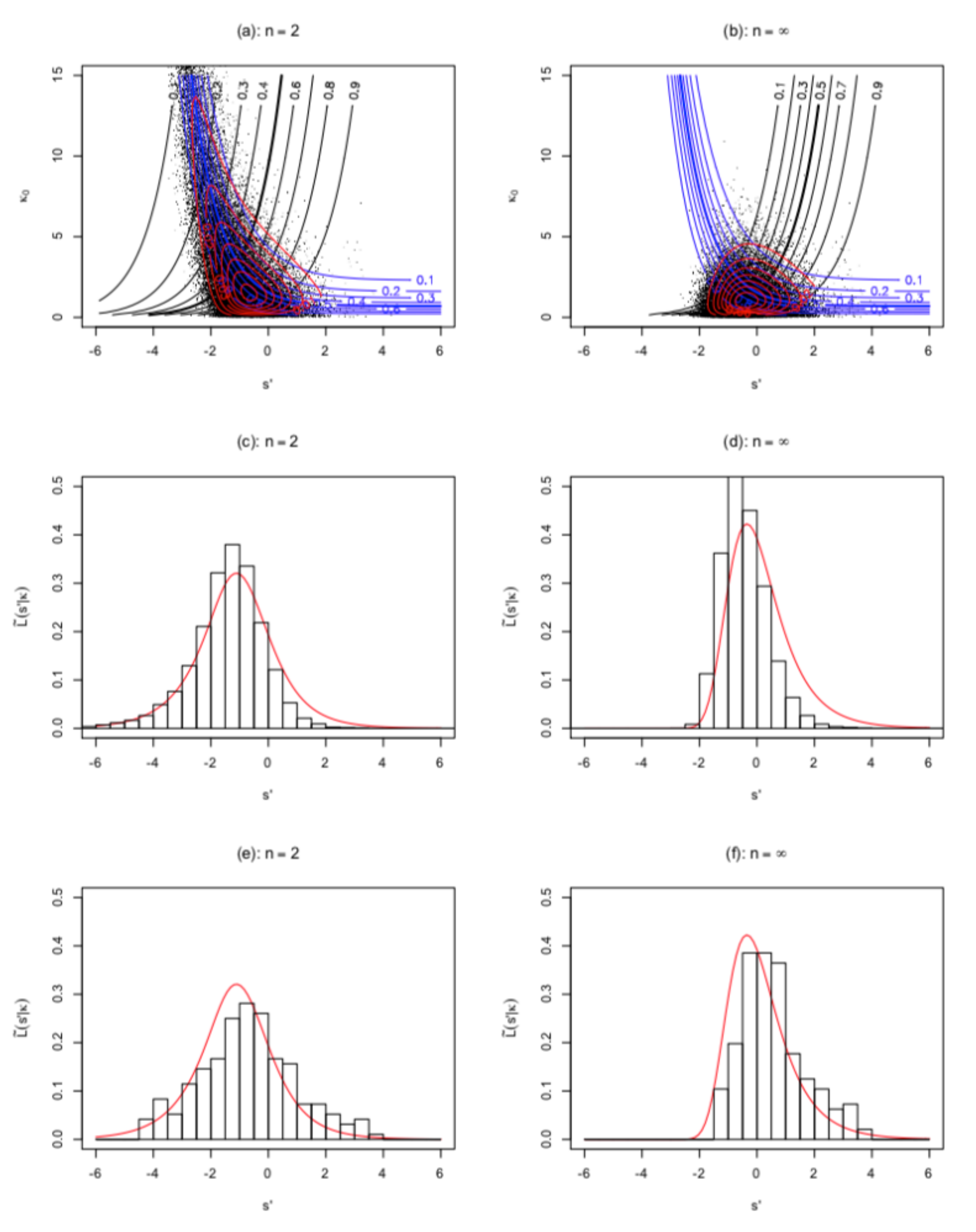}}
\caption{Plots (a) and (b): Contours of the asymptotic universal functions  $\tilde{u}_n(s', \kappa_0)$ (black), $\tilde{v}(s', \kappa_0)$ (blue) and 
likelihood surface $\tilde{\cal L}(s', \kappa_0)$ (red), together with a scatterplot of simulated data as described in the text.  
Plots (c) and (d): Asymptotic marginal likelihood $\tilde{\cal L}(s')$ together with a histogram of the data set of scaled and shifted coalescent 
times $s'$ with data corresponding to $\kappa < 50$ culled.  Plots (e) and (f): the same, but with more severe culling to include only trees 
corresponding to earlier starting times, as described in the text.} 
\label{fig:AnalyseSimulatedGWData}
\end{center}
\end{figure}

%
%
\section{Numerical simulation}
\label{sec:NumericalSimulation}

By exploiting the asymptotic universal functions arising from the $\kappa \rightarrow \infty$ limit we can compare the above theory with simulated 
data without needing to deal with explicit values of the final scaled population size $\kappa$.  

Plotted in Fig.~\ref{fig:AnalyseSimulatedGWData} are simulated data 
produced as follows.  A set of trees was generated, each starting with a single ancestor and evolving as a BGW process, with each 
parent in the process independently producing a Poisson number of offspring with $\lambda = \sigma^2 = 1.01$.  
Trees were generated until a dataset of 40000 family trees, each surviving 
for $t_{\rm sim} = 1000$ generations, was accumulated.  For each tree $n=2$ individuals were randomly chosen from the final generation and 
their ancestry traced to locate the time of their MRCA, and to record the population size at that time.  Each sample was taken from a 
separate tree to avoid correlations between the MRCAs of samples with a common history \citep{ball1990gene,Slatkin:1991uq}.  
For each family tree the time and population size corresponding to the MRCA of the entire population in the final generation 
was also recorded.  These coalescence times and population sizes were then transformed to the scaled quantities $s'$ and $\kappa_0$.  

For consistency with the $\kappa \rightarrow \infty$ limit, a small number of trees with final scaled population sizes $\kappa < 50$ 
were discarded to produce the scatter plots in Figures~\ref{fig:AnalyseSimulatedGWData}(a) and (b) and the histograms of $s'$ 
values in Figures~\ref{fig:AnalyseSimulatedGWData}(c) and (d).  It is clear that the right hand tails of the histograms fall short of the 
theoretical marginal likelihood curves.  To appreciate the cause of this discrepancy, note that the marginal likelihoods in Figures~\ref{fig:AnalyseSimulatedGWData}(c) and (d) drop to almost zero outside a finite range, $s'_{\rm lower} < s' <  s'_{\rm upper}$.  
However, for practical reasons, a numerical simulation is restricted to be of finite time, namely $t_{\rm sim}$ generations, and 
because of this most of the trees in the sample have not existed sufficiently far in the past since the initial founder to cover this 
entire range.  In an attempt to remedy this in Figures~\ref{fig:AnalyseSimulatedGWData}(e) and (f) we have also culled from 
the data all trees for which 
\begin{equation}
t_{\rm sim} \log\lambda < s'_{\rm upper} + \log\kappa.
\end{equation} 
This second culling is more severe, and removes the majority of generated trees.  Setting $s'_{\rm upper} = 5$, the dataset is reduced to 
trees satisfying $50 < \kappa <  141.2$, which reduces the original set to 192 trees.  Nevertheless, the 2 final histograms are in 
closer agreement with the right hand tails of the theoretical marginal likelihoods. 
%
%
\section{Conclusions}
\label{sec:Conclusions}

We have addressed the problem of establishing confidence regions for the scaled time $s = t \log \lambda$ since coalescence to a MRCA 
of a random sample and the scaled population population size $\kappa_0 = (2 m_0 \log\lambda)/\sigma^2$ at the time of the MRCA under 
the assumption that the current observed population is the result of a BGW process that has been evolving since some unknown time in 
the past.  Here $t$ is the number of time steps since the MRCA and $m_0$ is the corresponding population size.  The scalings are defined 
in terms of the mean $\lambda$ and variance $\sigma^2$ of the number of offspring per parent.  The mean $\lambda$ is assumed to be 
close to and slightly larger than 1, and $t$ is assumed to be large, so that the diffusion limit can be applied.  After scaling via Eq.~(\ref{sAndKDefns}) 
to implement the diffusion limit, the only free parameter in the problem of estimating $s$ and $\kappa_0$ is the currently observed scaled population 
$\kappa = 2 m \log\lambda/\sigma^2$, where $m$ is the currently observed physical population size.   

The approach differs markedly from earlier analyses of the MRCA in a BGW diffusion \citep{o1995genealogy,harris2019coalescent}
in that the BGW process is not initiated from a pre-specified point in the past.  Lacking information about the starting conditions 
is not a problem for tracing ancestry in WF-like models 
\citep{Slatkin:1991uq} which are inherently backward-looking in time.  In such models Kingman's coalescent is immediately applicable, and 
the concept of a probability density for the time since a MRCA is meaningful.  For a BGW process, however, causality runs forward in time, 
and to be on firm ground we have couched our results in terms of confidence regions.  Nevertheless some interpretation is in order.  

To return to the question posed in the opening paragraph of this paper, suppose we are confronted with a population which we are told 
has evolved through a BGW process.  Although the process is stochastic and its history unknown to us, 
one maintains an objective belief in the existence of a single realised path through the $(s, \kappa_0)$-plane leading 
to the current observed scaled population $\kappa$, and at some point that unknown path will have passed through co-ordinates corresponding 
to the MRCA of the current population.  Given that a BGW process is Markovian, the current population can be considered to be the result 
of a BGW process initiated from any point along the path, and in particular, from those co-ordinates corresponding to the MRCA.  
Our claim is that Eq.~(\ref{LikelihoodDensity}) calculated from the observed 
value of $\kappa$ defines a ``likelihood surface" in the sense that any subset $\Omega$ of the $(s, \kappa_0)$-plane  has a probability 
$\int_\Omega {\cal L}(s, \kappa_0 ; \kappa) ds d\kappa_0$ of containing the unknown but extant co-ordinates corresponding to the MRCA 
of the current population.  

Having made this claim, we mention two final caveats.  Firstly, note that Proposition~\ref{proposition1} is contingent on 
a conjecture that the likelihood surface over the $(u, v)$-plane is a uniform joint density.  Although the marginal densities in 
$u$ and $v$ are uniform, this is not necessarily the case for the joint density.  Indeed any density of the form $1 + \psi_1(u - v) + \psi_2(u + v)$ 
where $\psi_1$ and $\psi_2$ are periodic functions of a single variable with period 1 and which integrate to zero over $[0, 1]$ 
will have uniform marginal densities.
Without further mathematical proof, uniformity of the likelihood surface over the $(u, v)$-plane remains a conjecture.   

Secondly, if the likelihood surface ${\cal L}(s, \kappa_0; \kappa)$ is well defined, one feels it should be possible, at least in principle if not in practice, 
to devise a numerical simulation 
which will confirm its interpretation.  The simulation in Section~\ref{sec:NumericalSimulation} attempts to do this, but as we have seen, 
there is some ambiguity about how to match this particular simulation with theory.  Devising an appropriate simulation which is guaranteed to 
include the MRCA of the final population but does not impose any prior constraint on the population size at the time of the MRCA remains 
an open problem.  

%
%
\section*{Acknowledgements}

CB would like to thank Robert Griffiths, Robert Clark and Francis Hui for interesting and helpful discussions.  
%
%
\appendix
%
%
\section{Derivatives of $u_n(s, \kappa_0 | \kappa)$ and $v(s, \kappa_0 | \kappa)$}
\label{sec:DerivatesUAndV}

Derivatives of $u_n(s, \kappa_0 | \kappa)$ and $v(s, \kappa_0 | \kappa)$ are needed for the numerical determination of the likelihood 
surface defined by Eq.~(\ref{LikelihoodDensity}) and marginal likelihood calculation Eq.~(\ref{marginalLIntByParts}).  
Derivatives of $u_n(s, \kappa_0 | \kappa)$ are straightforward to calculate from Eqs.~(\ref{u_nResult}) and (\ref{wDef}).  
Using the identity \citep[][p376]{Abramowitz:1965sf}
\begin{equation}
\frac{d}{dz}\left(\frac{I_n(z)}{z^n}\right) = \frac{I_{n + 1}(z)}{z^n}, 
\label{besselDeriv}
\end{equation}
we obtain 
\begin{equation}
\frac{\partial u_n(s, \kappa_0 | \kappa)}{\partial s}  = -\frac{n!}{w^{n - 2}} \coth(\tfrac{1}{2}s) \Phi_n(w), \qquad 
\frac{\partial u_n(s, \kappa_0 | \kappa)}{\partial \kappa_0}  = \frac{n!}{\kappa_0 w^{n - 2}} \Phi_n(w),  
\label{uNDerivs}
\end{equation}
and
\begin{equation}
 \frac{\partial^2 u_n(s, \kappa_0 | \kappa)}{\partial s \partial \kappa_0} = 
 	\frac{- n!}{\kappa_0 w^{n - 2}} \left( \Phi_n(w) + w \Psi_n(w) \right) \coth (\tfrac{1}{2} s),  
\end{equation} 
where 
\begin{equation}
\Phi_n(w) = \frac{I_1(2w) I_{n + 1}(2w) - I_2(2w) I_n(2w)}{I_1(2w)^2}, 
\label{PhiDef}
\end{equation}
and
\begin{eqnarray}
\Psi_n(w) & = & \frac{1}{I_1(2w)^3} \left\{ I_1(2w)[I_1(2w) I_{n + 2}(2w) - I_3(2w)I_n(2w) ] \right. \nonumber \\
	&  & \quad -\, \left. 2I_2(2w) [I_1(2w) I_{n + 1}(2w) - I_2(2w)I_n(2w) ] \right\} .  
\end{eqnarray}
From Eq.~(\ref{u_InfResult}) the analogous derivatives for $n = \infty$ are 
\begin{equation}
\frac{\partial u_\infty(s, \kappa_0 | \kappa)}{\partial s}  = w^2 \coth(\tfrac{1}{2}s) \Phi_\infty(w), \qquad 
\frac{\partial u_\infty(s, \kappa_0 | \kappa)}{\partial \kappa_0}  = -\frac{w^2}{\kappa_0} \Phi_\infty(w),  
\label{uNDerivs}	 
\end{equation}
and
\begin{equation}
 \frac{\partial^2 u_\infty(s, \kappa_0 | \kappa)}{\partial s \partial \kappa_0} = 
 	\frac{w^2}{\kappa_0} \left( \Phi_\infty(w) + w \Psi_\infty(w) \right) \coth (\tfrac{1}{2} s),  
\label{du_nDsDkappa0}
\end{equation} 
where 
\begin{equation}
\Phi_\infty(w) = \frac{I_2(2w)}{I_1(2w)^2},  
\label{PhiDef}
\end{equation}
and
\begin{eqnarray}
\Psi_\infty(w) & = & \frac{I_1(2w) I_3(2w) - 2 I_2(2w)^2}{I_1(2w)^3} .  
\end{eqnarray}

Derivatives of $v(s, \kappa_0 | \kappa)$ are less straightforward.  The derivative with respect to $s$ can be calculated from the 
forward-Kolmogorov equation, Eq.~(\ref{fwdKolmog}), as follows; 
\begin{eqnarray}
\frac{\partial v(s, \kappa_0 | \kappa)}{\partial s}  & = & - \int_\kappa^\infty \frac{\partial f_K(\eta; s |  \kappa_0)}{\partial s} d\eta \nonumber \\
	& = & \int_\kappa^\infty 
		\left[\frac{\partial}{\partial \eta}(\eta f_K(\eta; s |  \kappa_0)) - \frac{\partial^2}{\partial \eta^2}(\eta f_K(\eta; s |  \kappa_0))\right] d\eta \nonumber \\
	& = & - \kappa f_K(\kappa; s |  \kappa_0) +  \frac{\partial}{\partial \kappa}\left(\kappa f_K(\kappa; s |  \kappa_0)\right) \nonumber \\   
	& = & w\left[ \left(\frac{1 - \kappa}{\kappa} - \frac{1}{e^s - 1}\right) I_1(2w) + \frac{w}{\kappa} I_2(2w) \right] 
					\exp\left\{-\frac{\kappa_0 e^s + \kappa}{e^s - 1}\right\},  \nonumber \\       
\label{dvDs}	
\end{eqnarray}
where Eqs.~(\ref{KDensity}), (\ref{wDef}) and (\ref{besselDeriv}) have been used in the last line.  
Evaluation of the derivative of $v$ with respect to $\kappa_0$ involves a numerical integration.  Some straightforward but lengthy algebra gives 
\begin{eqnarray}
\frac{\partial}{\partial \kappa_0} f_K(\kappa; s | \kappa_0) & = & - \frac{\delta(\kappa)}{1 - e^{-s}}  \exp\left\{-\frac{\kappa_0}{1 - e^{-s}}\right\} \nonumber \\
&  & + \, \frac{1}{2 \cosh s - 1} \exp\left\{ - \frac{\kappa_0 e^s + \kappa}{e^s - 1} \right\} \times \nonumber \\
& & \qquad\qquad \left[ I_2(2w) + \left(1 - \frac{\kappa_0 e^s}{e^s - 1} \right) \frac{I_1(2w)}{w} \right].    
\end{eqnarray}
Then from Eq.~(\ref{vDefn}) 
\begin{eqnarray}
\frac{\partial v(s, \kappa_0 | \kappa)}{\partial \kappa_0}  & = & \int_0^\kappa \frac{\partial}{\partial \kappa_0} f_K(\eta; s | \kappa_0) \, d\eta \nonumber \\
& = & - \frac{1}{1 - e^{-s}}  \exp\left\{-\frac{\kappa_0}{1 - e^{-s}}\right\} \nonumber \\
&  & + \, \frac{1}{2 \cosh s - 1} \int_0^\kappa \exp\left\{ - \frac{\kappa_0 e^s + \eta}{e^s - 1} \right\} \times \nonumber \\
& & \qquad\qquad \left[ I_2(2\omega) + \left(1 - \frac{\kappa_0 e^s}{e^s - 1} \right) \frac{I_1(2\omega)}{\omega} \right] \, d\eta,     
\label{dvDkappa0}	
\end{eqnarray}
where 
\begin{equation}
\omega = \frac{(\eta \kappa_0 e^s)^\frac{1}{2}}{e^s - 1}. 
\end{equation}

\section*{References}
\bibliographystyle{elsarticle-harv}\biboptions{authoryear}





\end{document}